\documentclass[runningheads]{llncs}

\usepackage{graphicx, amsmath, amssymb, amsfonts, tikz, mathtools, tikz-cd,epstopdf,verbatim, hyperref}
\usepackage{wrapfig,xspace,subcaption}
\usepackage[font=itshape]{quoting}
\usepackage{enumitem}
\usepackage{amsbsy,scalerel}
\usepackage[T1]{fontenc}
\usepackage{orcidlink}
\usetikzlibrary{automata, positioning, arrows}

\usepackage{amsthm}

\tikzset{
->, % makes the edges directed
node distance=3cm, % specifies the minimum distance between two nodes. Change if necessary.
every state/.style={thick, fill=gray!10}, % sets the properties for each ’state’ node
initial text=$ $, % sets the text that appears on the start arrow
}

\PassOptionsToPackage{bookmarks=false}{hyperref}

\theoremstyle{definition}

\newtheorem{axiom}{Axiom}

\graphicspath{{figures/}}

\allowdisplaybreaks

\usepackage{xcolor}

\newcommand{\setArg}[2]{\left\{ {#1} \,\,\middle|\,\, {#2} \right\}}

\newcommand{\cont}{\mathcal{C}}

\renewcommand{\paragraph}[1]{\noindent\textbf{#1}}

\renewcommand{\tt}{\texttt}

\newcommand{\balg}[0]{\mathbf{B}}
\newcommand{\calg}[0]{\mathbb{C}}

\begin{document}

\title{Composition and Merging of Assume-Guarantee Contracts Are Tensor Products}

\titlerunning{Tensor Products for Contracts}

\author{Inigo Incer\,\orcidlink{0000-0001-7933-692X}%\orcidID{0000-0001-7933-692X}
}
\authorrunning{Inigo Incer}
\institute{California Institute of Technology, Pasadena, CA, USA
%\\ \email{inigo@caltech.edu}\\
}

%%%%%%%%%%%%%%%%%% Content starts %%%%%%%%%%%%%%%%%%%%%%
\maketitle

\begin{abstract}
%This note shows that the operations of composition and merging of assume-guarantee contracts can be understood as tensor products 
%in suitably defined categories.
We show that the operations of composition and merging of contracts are part of the tensor product structure of the algebra of contracts.
\end{abstract}

\section{Introduction}
\label{jkhbgbf}
In the 2000s, Benveniste et al.~\cite{multViewpoint} proposed contracts to provide semantics to the component-based design methodologies common in the system industries. Instead of using detailed mathematical models to represent components, contract-based design can yield insight about a system using the assume-guarantee specifications of its components. The basic ingredients of analysis are the contracts and the algebraic operations on these objects.

The algebraic operations on contracts were discovered piecemeal.
The original paper on the topic~\cite{multViewpoint} discusses conjunction, disjunction, and composition.
Conjunction, which is the greatest lower bound induced from the contract order, was proposed a way to bring various aspects, dubbed \emph{viewpoints}, into the same contract. Disjunction was induced from the order on contracts as the least upper bound; its use in system design was not known for several years, until it was pointed out %J.M.~Jézéquel pointed out
that its role is in product lines---see \cite{Incer:EECS-2022-99}, page 21. This shows that logical considerations can sometimes be ahead of our practical understanding.
Composition was defined as the operation that yields the specification of a system from the specifications of components being interconnected. The composition formula was constructed from intuitive considerations and is reminiscent of the composition axiom proposed by M. Abadi and L. Lamport in the early 90s \cite{AbadiLamport93TPLS}. 
A fourth binary operation, merging, was proposed later \cite{contractMerging} on axiomatic grounds to fuse into a single contract multiple viewpoints that we want to hold simultaneously for a component. This shows that intuition can also be ahead of theory.

For each of the four basic operations just described, a new operation was discovered that optimally undoes it. These are dubbed \emph{adjoint} operations per their role in category theory---see Section 6.8 of \cite{Incer:EECS-2022-99}. For instance, composition builds a system from its constituent elements; its adjoint, quotient, finds the optimal specifications of objects that need to be added to a system to make it satisfy a goal. The following diagram provided in \cite{Incer:EECS-2022-99} shows how the various contract operations relate to each other.

\noindent
\scalebox{0.91}
{\scriptsize
$$
{
{
    \begin{tikzcd}[row sep = 0.1em, column sep = tiny,ampersand replacement=\&]
    \& \text{Conjunction } \land \arrow[ddd, "\text{Dual}\,\,"', Leftrightarrow] 
    \arrow[rrr, "\text{Right adjoint}", Rightarrow]
    \& [2.9em] \& \& \text{{{Implication}} } \rightarrow \\ 
    \text{Order} \arrow[ur, Rightarrow] \arrow[ddr, Rightarrow] \& \& \\
    \\
    \& \text{Disjunction } \lor \arrow[rrr, "\text{Left adjoint}", Rightarrow] \& \& \&
    \text{{{Coimplication}} } \nrightarrow \arrow[uuu, "\,\,\text{Dual}"', Leftrightarrow] 
    \end{tikzcd}
    %%%%%%
    \begin{tikzcd}[row sep = 0.1em, column sep = tiny,ampersand replacement=\&]
    \& \text{Composition} \parallel \text{\;}\arrow[ddd, "\text{Dual}\,\,"', Leftrightarrow] 
    \arrow[rrr, "\text{Right adjoint}", Rightarrow]
    \& [1.5em] \& \& \text{ {{Quotient}} } / \\ 
    \text{Axiom} \arrow[ur, Rightarrow] \arrow[ddr, Rightarrow] \& \& \\
    \\
    \& \text{{{Merging}} } \bullet \arrow[rrr, "\text{Left adjoint}", Rightarrow] \& \& \&
    \text{{{Separation}} } \div \arrow[uuu, "\,\,\text{Dual}"', Leftrightarrow]
    \end{tikzcd}
}}$$}

\vspace{1mm}

This diagram shows that two operations are obtained from the notion of order, and two by axiom: composition and merging. The operations of conjunction, disjunction, and its adjoints provide contracts with significant algebraic structure \cite{incer2023adjunction}. Given the practical, widespread use of merging and composition, we wonder whether these operations have an algebraic interpretation, as well.

\textbf{Contributions.} The purpose of this note is to study the algebraic roles of the operations of merging and composition in the theory of contracts. It is shown that these operations can be understood as tensor products. This interpretation enables the identification of these operations without an appeal to axiom.

\textbf{Outline.} Section~\ref{kjcnhkd} contains basic definitions about contracts. This paper makes no contributions in this section.
Section~\ref{kljqcnhkgd} investigates the tensor product of contracts. Except where explicitly indicated, this section contains new material. We discuss these results in Section \ref{knjgbksjgbj}.
%Conclusions are given in Section~\ref{kcsdhhj}.

\section{Assume-guarantee contracts}
\label{kjcnhkd}

To define contracts, we assume we have access to a language used to express properties. We will require this language to have Boolean semantics (i.e., Boolean connectives plus the excluded middle), and as a result, we shall refer to it as $\balg$. The notation in this section follows closely ~\cite{Incer:EECS-2022-99} and~\cite{incer2023algebraic}. 

\begin{definition}
Given a Boolean algebra $B$, its contract algebra $\calg(\balg)$ is given by
\begin{equation*}
\label{knxqkjbq}
\calg(\balg) = \setArg{(a,g) \in \balg^{\text{\normalfont op}} \times \balg}{a \lor g = 1_\balg}.
\end{equation*}
\end{definition}

Some observations about this definition:
\begin{itemize}
\item A contract is a pair $(a,g)$ of \emph{assumptions} and \emph{guarantees}. The semantics of a contract, and the way the contract operations were originally defined, was by understanding a contract as a specification that requires a component to satisfy the contract's guarantees $g$ when the environment in which it operates satisfies the contract's assumptions $a$.
\item The condition $a \lor g = 1_{\balg}$ in the above definition is there to give to the pairs $(a,g)$ the semantics of assume-guarantee reasoning. Note that if the assumptions $a$ do not hold, $g$ must be asserted; in other words, this condition says that if the assumptions do not hold, the guarantees hold ipso facto.
\item $\balg^{\text{op}}$ means the opposite of $\balg$, the algebra in which $\land$ has been switched with $\lor$, and the top with the bottom elements. This notation is there to embed in the definition \eqref{knxqkjbq} of the contract algebra of $\balg$ the following partial order:
$$
(a,g) \le (a', g') \text{ iff }
g \le g' \text{ and }
a' \le a.
$$
Observe this definition matches intuition. It says that a contract $(a,g)$ is stricter---or a \emph{refinement}---of contract $(a',g')$ if the guarantees of the stricter contract are stricter than those of the more relaxed contract, and the assumptions the stricter contract makes on its environment are more relaxed than those that the more relaxed contract makes. 
\end{itemize}

$\calg(\balg)$ is a bounded partial order:

\begin{corollary}\label{jbhgsgqjk}
The contracts $0 \coloneq (1_\balg, 0_\balg)$ and $1 \coloneq (0_\balg, 1_\balg)$ are, respectively, the bottom and top of $\calg(\balg)$.
\end{corollary}

\subsection{Conjunction and disjunction}

The definition of contracts, with its embedded notion of order, immediately generates two binary operations: conjunction and disjunction, defined as the greatest lower bound and least upper bound for arbitrary pairs of contracts, respectively. Note that these two operations come perforce from the definition of contracts.

\begin{proposition}
Let $(a,g), (a',g') \in \calg(\balg)$. The greatest lower bound (or conjunction) and the least upper bound (or disjunction) of these contracts are, respectively, given by
\begin{equation*}
(a,g) \land (a',g') = (a \lor a', g\land g') \text{ and }
(a,g) \lor (a',g') = (a \land a', g\lor g').
\end{equation*}
\end{proposition}
\begin{proof}
$(a'',g'') \le (a,g)$ and $(a'',g'') \le (a',g')$
means that
$g'' \le g$ and $a \le a''$ and $g'' \le g'$ and $a' \le a''$.
This is equivalent to
$g'' \le g \land g'$ and $a \lor a' \le a''$, which means that
$(a'',g'') \le (a \lor a', g \land g')$. This proves the expression for conjunction. A similar procedure leads to the formula for disjunction.
\end{proof}

\subsection{Axioms for composition and merging}

Two types of models of $\balg$ are associated with contracts over $\balg$. They go by the names of environments and implementations.
\begin{definition}
    An \emph{environment} of a contract $(a,g) \in \calg(\balg)$ is a model of $a$, and an \emph{implementation} of this contract is a model of $g$.
\end{definition}

Given contracts $\cont$ and $\cont'$, two operations are defined by axiom:

\begin{axiom}
The \emph{merger} of contracts $\cont$ and $\cont'$, denoted $\cont \bullet \cont'$, is the most relaxed contract $\cont_m$ satisfying these two requirements:
\begin{enumerate}
    \item An implementation of $\cont_m$ is an implementation of both $\cont$ and $\cont'$.
    \item An environment of $\cont$ that is also an environment of $\cont'$ is an environment of $\cont_m$.
\end{enumerate}
\end{axiom}

To define the axiom for the composition of contracts, we assume that models of $\balg$ have a well-defined composition operation. For details, we refer the reader to Chapter 5 of \cite{BenvenisteContractBook} and \cite{hypercontractsNasa}.

\begin{axiom}
The \emph{composition} of contracts $\cont$ and $\cont'$, denoted $\cont \parallel \cont'$, is the most refined contract $\cont_c$ satisfying these two requirements:
\begin{enumerate}
    \item The composition of an implementation of $\cont$ with an implementation of $\cont'$ yields an implementation of $\cont_c$.
    \item The composition of an environment of $\cont_c$ with an implementation of one of the contracts being composed is an environment of the other contract being composed.
\end{enumerate}
\end{axiom}

The application of these two axioms leads to the following closed form expressions for composition and merging:

\begin{corollary}
Given $(a,g), (a',g') \in \calg(\balg)$,
the closed form expressions to compute merging and composition are
\begin{align}
\label{bkjfgbkjf}
(a,g) \bullet (a',g') &= \left( a \land a', (a \land a') \to (g \land g') \right) \text{ and}\\
\label{vdhqgngbq}
(a,g) \parallel (a',g') &= \left( (g \land g') \to (a \land a'), g \land g'  \right).
\end{align}
\end{corollary}
Here, $\to$ stands for the implication operator of Boolean algebras.
Observe the difference in the way we arrived at the expressions to compute conjunction/disjunction and composition/merging. The former were obtained by algebraic means, i.e., as the LUB and GLB of the refinement order, while the latter were obtained by axiom. The purpose of this note is to understand the role that merging and composition play in the algebra of contracts. In particular, we would like to understand whether their closed form expressions can be obtained by algebraic means. It turns out they are related to tensor products.

%Contracts have other two operations which are very useful in practice: contract merging and composition. These are defined \emph{by axiom}.

\begin{comment}
\subsection{Duality}

We recall that contracts have an involution named reciprocal whose behavior is given by
$$
(a,g)^{-1} = (g,a).
$$
This operation gives rise to a relation between binary operations in a contract algebra. We say that operations $\circ$ and $\diamond$ are dual if they satisfy
$$
\left(\cont \circ \cont'\right)^{-1} =
\cont^{-1} \diamond (\cont')^{-1}.
$$
In ?, it is shown that several operations on contracts are duals of each other.

\end{comment}

\section{Tensor products of contracts}
\label{kljqcnhkgd}

In this section, we study a notion of tensor product of contracts with the tools so far developed. First, we discuss the notion of a tensor product over bimodules of monoids. Then we apply these notions to contracts.

\subsection{Tensor products of bimodules over monoids}

Tensor products are normally defined over ring modules. We will need to modify standard definitions slightly to obtain tensor products over structures with actions of monoids. A comprehensive reference on tensor products is Chapter 11 of \cite{h113}.

We recall that a monoid $(B, \cdot, 1_B)$ is a semigroup $(B,\cdot)$ with an identity element $1_B$. Let $(M, +, 0_M)$ be a commutative monoid, i.e., a monoid with a commutative operation $+$.
We say that $M$ is a $B$-bimodule if it is equipped with operations $\cdot \colon M, B \to M$ and $\cdot\colon B, M \to M$ satisfying the following properties for all $x,y \in B$ and $m,m' \in M$:
\begin{itemize}
    \item $1_B \cdot m = m = m \cdot 1_B$
    \item $(xy) \cdot m = x \cdot (y\cdot m)$
    \item $ m\cdot (xy) = (m \cdot x) \cdot y$
    \item $(x \cdot m) \cdot y = x \cdot (m \cdot y)$
    \item $x \cdot (m + m') \cdot y = (x \cdot m \cdot y) + (x \cdot m' \cdot y)$
\end{itemize}

The operation that applies an element of $B$ on the left is called a \emph{left action}, and the operation that applies it on the right a \emph{right action}.

Suppose that $M$, $N$, and $O$ are $B$-bimodules. A $B$-linear map $f\colon M \to O$ is a map satisfying the following identities for all $m,m' \in M$ and $x,y \in B$:
\begin{itemize}
    \item $f(m + m') = f(m) + f(m')$
    \item $x \cdot f(m) \cdot y = f(x \cdot m \cdot y)$
\end{itemize}
A $B$-bilinear map $f\colon M, N \to O$ is a map satisfying the following identities for all $m,m' \in M$, $n,n'\in N$, and $x,y \in B$:
\begin{itemize}
    \item $f(m + m', n) = f(m, n) + f(m', n)$
    \item $f(m, n+n') = f(m, n) + f(m, n')$
    \item $x \cdot f(m, n) \cdot y = f(x \cdot m \cdot y, n) = f(m, x \cdot n \cdot y)$
\end{itemize}

A tensor product of $B$-bimodules $M$ and $N$ is a $B$-bimodule $M \otimes N$ together with a $B$-bilinear map $\tau\colon M,N \to M \otimes N$ satisfying the following property: for every $B$-bilinear map $f \colon M,N \to O$, there is a unique $B$-linear map $M\otimes N \to O$ making the following diagram commute:
$$
\begin{tikzcd}
M,N \arrow[r, "f"] \arrow[d, "\tau"] & O \\
M \otimes N \arrow[ur, ""', dashed]
\end{tikzcd}
$$

\subsection{The tensor product of conjunctive monoids of contracts}
\label{jbacgja}

\newcommand{\andmon}{\calg_\land^M}
\newcommand{\ormon}{\calg_\lor^M}
\newcommand{\parmon}{\calg_\parallel^M}
\newcommand{\mermon}{\calg_\bullet^M}

\newcommand{\boolandmon}[0]{\mathbf{M}_{\land}}
\newcommand{\boolormon}[0]{\mathbf{M}_{\lor}}

\newcommand{\boolcat}[0]{\textbf{Bool}}
\newcommand{\moncat}[0]{\textbf{Mon}}

Let $\boolcat$ be the category of Boolean algebras, and $\moncat$ that of monoids. We consider the
functors $\boolandmon, \andmon \colon \boolcat \to \moncat$ defined as follows:
\begin{align*}
    &\balg \xmapsto{\boolandmon} \left( |\balg|, \land, 1_\balg \right) 
    &&f \xmapsto{\boolandmon} f \quad \text{and} \\
    %%%
    &\balg \xmapsto{\andmon} \left( |\calg(\balg)|, \land, 1 \right)
    &&f \xmapsto{\andmon} f\times f,
\end{align*}
where $\balg$ and $f$ are, respectively, an object and a morphism in $\boolcat$. In other words, $\boolandmon(\balg)$ is the monoid obtained by only remembering the conjunction operation and the distinguished element $1_\balg$ of the Boolean algebra $\balg$; similarly, $\andmon(\balg)$ is the monoid obtained by only remembering the conjunction operation and the distinguished element $1$ of the contract algebra $\calg(\balg)$---see Corollary~\ref{jbhgsgqjk}. Contract monoids are studied in Section 6.10.1 of \cite{Incer:EECS-2022-99}.
For ease of syntax, throughout this paper we will say $\balg$-bimodules and $\balg$-linear/bilinear maps instead of $\boolandmon(\balg)$-bimodules, etc.

We now turn the monoid $\andmon(\balg)$ into a $\balg$-bimodule through the following operations:
\begin{align*}
&x \cdot (a, g) \coloneq ((x \land a), x \to g) && \text{and}\\
&(a,g) \cdot x \coloneq (a, a \to (x \land g)) && \left(x \in \balg, (a,g) \in \calg(\balg)\right).
\end{align*}

This left action was defined in Section 6.11 of \cite{Incer:EECS-2022-99}. The right action is, to the best of our knowledge, new.
First, we observe that $(x \land a) \lor (x \to g) = 1_\balg = a \lor \left( a \to (x \land g) \right)$. This means that $x \cdot (a, g)$ and $(a, g) \cdot x$ are elements of $\calg(\balg)$. Second, the operations just defined have the following interpretation:
given a contract $\cont \in \calg(\balg)$, the left operation $x \cdot \cont$ has the effect of adding to the contract $\cont$ an additional assumption $x$. % while keeping its guarantees (in the context of its assumptions) fixed.
Similarly, the contract $\cont \cdot x$ has the same assumptions as $\cont$, but it has additional guarantees $x$. In other words, $x\cdot \cont$ adds $x$ to the assumptions of $\cont$, and $\cont \cdot x$ adds $x$ to its guarantees.

\begin{proposition}\label{kwjbfgc}
Under the left and right operations just defined, $\andmon(\balg)$ is a $\balg$-bimodule.
\end{proposition}
\begin{proof}
Let $x,y \in \balg$ and $(a,g), (a',g') \in \calg(\balg)$. We verify the conditions of a bimodule stated in Section \ref{jbacgja}:
\begin{itemize}
\item $1_\balg \cdot (a,g) = ((1_\balg \land a), 1_\balg\to g) = (a,g) = (a, a \to (1_\balg \land g)) = (a,g) \cdot 1_\balg$.
\item $(x\land y) \cdot (a,g) = x\cdot \left( y \cdot (a,g)\right)$, as show in the proof of Proposition 6.11.1 of \cite{Incer:EECS-2022-99}.
\item $\begin{aligned}[t]
(a,g) \cdot (x\land y) & = \left( a, a \to (g \land (x \land y)) \right) \\ &=
\left( a, a \to ((g \land x) \land y) \right) \\ &=
\left( a, a \to (g \land x) \right) \cdot y \\ &=
\left((a,g) \cdot x\right) \cdot y.
\end{aligned}
$
\item $
\begin{aligned}[t]
\left(x \cdot (a,g)\right) \cdot y & =
\left( (x \land a), x \to g \right) \cdot y \\ &=
\left( (x \land a), (x \land a) \to \left(y \land (x \to g)\right) \right) \\ &=
\left( (x \land a), (x \land a) \to \left(y \land g\right) \right) \\ &=
\left( (x \land a), x \to \left(a \to \left(y \land g\right)\right) \right) \\ &=
x \cdot \left( a, a \to \left(y \land g\right) \right) \\ &=
x\cdot \left( (a,g) \cdot y\right).
\end{aligned}
$
\item We will use the following identity:
\begin{equation}\label{jkwbhfgjk}
(a \lor a') \to (g \land g') = g \land g' = (a \to g) \land (a' \to g').
\end{equation}
We can now verify the property:
\begin{align*}
x \cdot & \left( (a,g) \land (a',g') \right) \cdot y =
x \cdot \left( (a \lor a',g \land g') \right) \cdot y \\&=
x \cdot \left( a \lor a', (a \lor a') \to \left( y \land g \land g' \right) \right) \\ &=
%%%%%%%%%%%%%%
\left( (x\land a) \lor (x \land a'), x \to \left( (a \lor a') \to \left( y \land g \land g' \right) \right) \right) \\ &\stackrel{\eqref{jkwbhfgjk}}{=}
%%%%%%%%%%
\left( (x\land a) \lor (x \land a'), x\to
\left(
\begin{aligned}
& \left( a \to ( y \land g ) \right) \land \\ &
\left( a' \to ( y \land g' ) \right)
\end{aligned}
\right)
\right) \\ &=
%%%%%%%%%%%%%%%
\left( (x\land a) \lor (x \land a'), 
\left(
\begin{aligned}
& \left( (x\land a) \to ( y \land g ) \right) \land \\ &
\left( (x\land a') \to ( y \land g' ) \right)
\end{aligned}
\right)
\right) \\ &=
%%%%%%%%%%%%
\left( x\land a, 
(x\land a) \to ( y \land g )
\right) \land
\left( x\land a', 
(x\land a') \to ( y \land g' )
\right)
\\ &=
%%%%%%%%%%%%
(x \cdot (a,g) \cdot y) \land (x \cdot (a',g') \cdot y).
\qedhere
\end{align*}
\end{itemize}
\end{proof}

\begin{theorem}
$\andmon(\balg)$ together with the map $(\cdot)\bullet(\cdot)\colon \andmon(\balg), \andmon(\balg) \to \andmon(\balg)$ given by \eqref{bkjfgbkjf} is the tensor product $\andmon(\balg) \otimes \andmon(\balg)$.
\end{theorem}
\begin{proof}
We know from Proposition \ref{kwjbfgc} that $\andmon(\balg)$ is a $\balg$-bimodule.
First we verify that $(\cdot)\bullet(\cdot)$ is a $\balg$-bilinear map. Observe that this operation is commutative, so we only have to check additivity/linearity for one argument. We carry out the following verifications:
\begin{itemize}
\item Let $\cont, \cont', \cont'' \in \andmon(\balg)$. Then $(\cont \land \cont') \bullet \cont'' = (\cont \bullet \cont'') \land (\cont' \bullet \cont'')$ from Proposition 6.10.4 of \cite{Incer:EECS-2022-99}.
\item Let $x,y \in \balg$ and $(a,g), (a',g') \in \andmon(\balg)$. We have
\begin{align*}
x \cdot \left( (a,g) \bullet (a',g') \right) \cdot y &= 
\left( x\land a\land a', (x\land a\land a') \to (y\land g\land g') \right) \\ &=
\left( x\land a, (x\land a) \to (y\land g) \right) \bullet (a', g') \\ &=
\left( x \cdot (a,g) \cdot y \right) \bullet (a', g'). 
\end{align*}
\end{itemize}
We have shown that merging is a $\balg$-bilinear map $\andmon(\balg), \andmon(\balg) \to \andmon(\balg)$. Now suppose $N$ is a $\balg$-bimodule and that $f\colon \andmon(\balg), \andmon(\balg) \to N$ is a $\balg$-bilinear map.
Let $e \coloneq (1_\balg,1_\balg) \in \andmon(\balg)$.
Define the map $\hat f \colon \andmon(\balg) \to \andmon(\balg)$ by
$$
\hat f (\cont) = f(e, \cont).
$$
Since $f$ is $\balg$-bilinear, $\hat f$ is $\balg$-linear. Let $(a,g), (a',g') \in \andmon(\balg)$. We observe that
\begin{align*}
f((a,g),(a',g')) &= f(a\cdot e\cdot g, (a',g')) = f( e, a\cdot(a',g')\cdot g) \\ &=
f( e, (a \land a', (a\land a') \to (g\land g'))) = f(e, (a,g)\bullet(a',g')) \\ &= \hat f ((a,g)\bullet(a',g')).
\end{align*}
This means that there exists a $\balg$-linear map $\hat f$ making the following diagram commute:
\begin{equation}\label{jhbdcgqjd}
\begin{tikzcd}
\andmon(\balg),\andmon(\balg) \arrow[r, "f"] \arrow[d, "(\cdot)\bullet(\cdot)"] & N \\
\andmon(\balg) \arrow[ur, "\hat f"']
\end{tikzcd}
\end{equation}
Suppose there is a second $\balg$-linear map $h\colon \andmon(\balg) \to \andmon(\balg)$ such that
$$
\hat f (\cont \bullet \cont') = h (\cont \bullet \cont')
$$
for all $\cont, \cont' \in \andmon(\balg)$. By setting $\cont' = e$ in this expression, we obtain that $\hat f = h$. It follows that there is a unique map $\hat f$ making \eqref{jhbdcgqjd} commute. The statement of the theorem follows.
\end{proof}

\subsection{The tensor product of disjunctive monoids of contracts}
\label{khjkbcd}

In this section, we consider the disjunctive contract monoid functor $\ormon \colon \boolcat \to \moncat$ (see Proposition 6.10.1 of \cite{Incer:EECS-2022-99}) defined as follows:
\begin{align*}
    &\balg \xmapsto{\ormon} \left( |\calg(\balg)|, \lor, 0 \right)
    &&f \xmapsto{\ormon} f\times f.
\end{align*}
Proposition 6.10.2 of \cite{Incer:EECS-2022-99} shows that $\ormon(\balg)$ and $\andmon(\balg)$ are isomorphic \emph{as monoids} via the map $(\cdot)^{-1} \colon \calg(\balg) \to \calg(\balg)$ defined by
$$
(a,g)^{-1} = (g,a).
$$
%That is, $(\cdot)^{-1}$ is an isomorphism of the monoids $\ormon(\balg)$ and $\andmon(\balg)$.
Let $\cont = (a,g) \in \ormon(\balg)$ and $\cont' = (g,a) \in \andmon(\balg)$. There are unique left and right actions of $\boolandmon(\balg)$ on $\ormon(\balg)$ that make $(\cdot)^{-1}$ a linear map:
\begin{align*}
x\cdot (a,g) &=
x \cdot \cont =
x \cdot (\cont')^{-1} \coloneq
\left(x\cdot \cont'\right)^{-1} =
\left(x\cdot (g,a) \right)^{-1} \\ &=
\left( x\land g, x \to a \right)^{-1} =
(x \to a, x\land g) \\
%%%%%
(a,g)\cdot x &=
\cont \cdot x =
(\cont')^{-1} \cdot x \coloneq
(\cont' \cdot x)^{-1} =
((g,a) \cdot x)^{-1} \\ &=
(g,g \to (x\land a))^{-1} =
(g \to (x\land a), g).
\end{align*}
This means that $(\cdot)^{-1}$ is an isomorphism of $\ormon(\balg)$ and $\andmon(\balg)$ \emph{as $\balg$-bimodules}. Now we can prove our main result.

\begin{theorem}
    $\ormon(\balg)$ together with the map $(\cdot)\parallel(\cdot)\colon \ormon(\balg), \ormon(\balg) \to \ormon(\balg)$ given by \eqref{vdhqgngbq} is the tensor product $\ormon(\balg) \otimes \ormon(\balg)$.
    \end{theorem}
\begin{proof}
Consider the $\balg$-bimodule $N$ and the $\balg$-bilinear map $f\colon \ormon(\balg), \ormon(\balg) \to N$.
Since $(\cdot)^{-1}$ is an isomorphism, for every such map $f$, there exist
a unique $\balg$-bilinear map
$\tilde f\colon \andmon(\balg), \andmon(\balg) \to N$ such that
$$
f (\cont, \cont') = \tilde f\left(\cont^{-1},\left( \cont'\right)^{-1}\right).
$$
Per Proposition~\ref{kwjbfgc}, there is a unique 
$\balg$-linear map $\hat f \colon \andmon(\balg) \to N$ such that
$$
\tilde f (\cont, \cont') = \hat f (\cont \bullet \cont').
$$
Per the isomorphism $(\cdot)^{-1}$, there is a unique $\balg$-linear map $\breve f\colon \ormon(\balg) \to N$ such that
$$
\hat f = \breve f (\cdot)^{-1}.
$$
Putting this reasoning together, we have shown that for every $\balg$-bilinear map $f$, there is a unique $\balg$-linear map $\breve f$ such that
\begin{align*}
f(\cont, \cont') &=
%\tilde f \left(\cont^{-1},\left( \cont'\right)^{-1}\right) = 
%\hat f \left( \cont^{-1} \bullet \left( \cont'\right)^{-1} \right) =
%\hat f \left(\left( \cont^{-1} \bullet \left( \cont'\right)^{-1} \right)^{-1}\right) \\ &=
\breve f \left(\left( \cont^{-1} \bullet \left( \cont'\right)^{-1} \right)^{-1}\right) =
\breve f \left( \cont \parallel \cont' \right),
\end{align*}
where the second equality is due to the duality of merging and composition---see Section 6.4 of \cite{Incer:EECS-2022-99}. In other words, for every $\balg$-bilinear map $f$, there is a unique $\balg$-linear map $\breve f$ making the following diagram commute:
\begin{equation*}
\begin{tikzcd}
\ormon(\balg),\ormon(\balg) \arrow[r, "f"] \arrow[d, "(\cdot)\parallel(\cdot)"] & N \\
\ormon(\balg) \arrow[ur, "\breve f"']
\end{tikzcd},
\end{equation*}
which proves the theorem.
\end{proof}

\section{Discussion and concluding remarks}
\label{knjgbksjgbj}

We have shown that the operations of merging and composition of contracts, which up to now have been defined by axiom, belong naturally in the tensor product structure of contracts.
The diagram we presented in Section~\ref{jkhbgbf}---showing all known binary contract operations---can now be recast as shown below. The diagonals of the bottom square correspond to the preheap\footnote{Preordered heaps, or preheaps, are studied in \cite{EPTCS326.14}.} identities described in Section 6.9 of \cite{Incer:EECS-2022-99}.

%\noindent
%\scalebox{0.99}
%{%\scriptsize
\newcommand{\boxlength}[0]{0.15\textwidth}
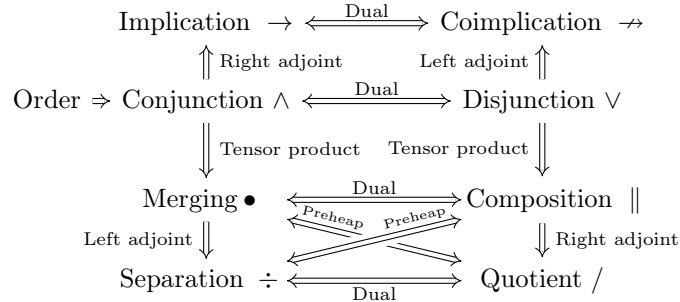
$$
{
{
    \begin{tikzcd}[row sep = 0.2em, column sep = tiny,ampersand replacement=\&]
    \& \text{Implication } \to \arrow[ddd, "\,\,\text{Right adjoint}\,\,", Leftarrow] 
    \arrow[rrr, "\text{Dual}", Leftrightarrow]
    \& [2.9em] \& \& \text{{{Coimplication}} } \nrightarrow \\ 
    \text{ } \& \& \\
    \\
    \text{Order} %\arrow[ur, Rightarrow]
    \arrow[r, Rightarrow]
    \& \text{Conjunction } \land \arrow[rrr, "\text{Dual}", Leftrightarrow] \arrow[ddd, "\,\,\text{Tensor product}\,\,", Rightarrow] \& \& \&
    \text{{{Disjunction}} } \lor \arrow[uuu, "\,\,\text{Left adjoint}\,\,", Rightarrow] 
    \arrow[ddd, "\,\,\text{Tensor product}\,\,"', Rightarrow]
    \\
    \& \& \& \& \text{}
    \\
    \& \& \& \& \text{}
    \\
    %\end{tikzcd}
    %%%%%%
    %\begin{tikzcd}[row sep = 0.1em, column sep = tiny,ampersand replacement=\&]
    \& \mathmakebox[\boxlength]{\text{Merging} \bullet \text{\;}} \arrow[ddd, "\,\,\text{Left adjoint}\,\,"', Rightarrow] 
    \arrow[rrr, "\text{Dual}", Leftrightarrow]
    \arrow[dddrrr, "\text{Preheap}" {near start, sloped, font=\tiny}, Leftrightarrow]
    \& [1.5em] \& \& \mathmakebox[\boxlength]{\text{\;\; {{Composition}} } \parallel} 
    \arrow[dddlll, "\text{Preheap}" {near start, sloped, font=\tiny}, Leftrightarrow, crossing over]
    \\ 
    \text{} %\arrow[ur, Rightarrow] \arrow[ddr, Rightarrow]
    \& \& \\
    \\
    \& \mathmakebox[\boxlength]{\text{{{Separation}} } \div \text{\;}} \arrow[rrr, "\text{Dual}"', Leftrightarrow] \& \& \&
    \mathmakebox[\boxlength]{\text{{{Quotient}} } /} \arrow[uuu, "\,\,\text{Right adjoint}\,\,"', Leftarrow]
    \end{tikzcd}
}}$$%}

%\section{Concluding remarks}
%\label{kcsdhhj}

\subsection*{Acknowledgements}

This work was supported by ASEE and NSF through the eFellows postdoctoral program.

\bibliographystyle{style/splncs04}
\bibliography{support/references}

\begin{thebibliography}{10}
\providecommand{\url}[1]{\texttt{#1}}
\providecommand{\urlprefix}{URL }
\providecommand{\doi}[1]{https://doi.org/#1}

\bibitem{AbadiLamport93TPLS}
Abadi, M., Lamport, L.: Composing specifications. ACM Transactions on
  Programming Languages and Systems  \textbf{15}(1),  73--132 (January 1993)

\bibitem{multViewpoint}
Benveniste, A., Caillaud, B., Ferrari, A., Mangeruca, L., Passerone, R.,
  Sofronis, C.: Multiple viewpoint contract-based specification and design. In:
  Formal Methods for Components and Objects, $6^{th}$ International Symposium
  (FMCO 2007), Lecture Notes in Computer Science, vol.~5382, pp. 200--225.
  Springer Verlag, Berlin Heidelberg (2008). \doi{10.1007/978-3-540-92188-2}

\bibitem{BenvenisteContractBook}
Benveniste, A., Caillaud, B., Nickovic, D., Passerone, R., Raclet, J.B.,
  Reinkemeier, P., Sangiovanni-Vincentelli, A.L., Damm, W., Henzinger, T.A.,
  Larsen, K.G.: Contracts for system design. Foundations and
  Trends$^{\text{\scriptsize{\textregistered}}}$\hspace{-.3em} in Electronic
  Design Automation  \textbf{12}(2-3),  124--400 (2018).
  \doi{10.1561/1000000053}

\bibitem{Incer:EECS-2022-99}
Incer, I.: The Algebra of Contracts. Ph.D. thesis, EECS Department, University
  of California, Berkeley (May 2022),
  \url{http://www2.eecs.berkeley.edu/Pubs/TechRpts/2022/EECS-2022-99.html}

\bibitem{incer2023adjunction}
Incer, I.: An adjunction between {B}oolean algebras and a subcategory of
  {S}tone algebras. arXiv:2309.04135  (2023)

\bibitem{incer2023algebraic}
Incer, I., Benveniste, A., Sangiovanni-Vincentelli, A.: Some algebraic aspects
  of assume-guarantee reasoning. arXiv:2309.08875  (2023)

\bibitem{hypercontractsNasa}
Incer, I., Benveniste, A., Sangiovanni-Vincentelli, A.L., Seshia, S.A.:
  Hypercontracts. In: Deshmukh, J., Havelund, K., Perez, I. (eds.) NASA Formal
  Methods. pp. 674--692. Springer International Publishing, Cham (2022).
  \doi{10.1007/978-3-031-06773-0\_36}

\bibitem{EPTCS326.14}
Incer, I., Mangeruca, L., Villa, T., Sangiovanni-Vincentelli, A.L.: The
  quotient in preorder theories. In: Raskin, J.F., Bresolin, D. (eds.)
  {Proceedings 11th International Symposium on} Games, Automata, Logics, and
  Formal Verification, {Brussels, Belgium, September 21-22, 2020}. Electronic
  Proceedings in Theoretical Computer Science, vol.~326, pp. 216--233. Open
  Publishing Association, Brussels, Belgium (2020). \doi{10.4204/EPTCS.326.14}

\bibitem{contractMerging}
Passerone, R., Incer, I., Sangiovanni-Vincentelli, A.L.: Coherent extension,
  composition, and merging operators in contract models for system design. ACM
  Trans. Embed. Comput. Syst.  \textbf{18}(5s) (Oct 2019).
  \doi{10.1145/3358216}

\bibitem{h113}
Wodzicki, M.: Notes on Category Theory. Lecture notes for Math H113---UC
  Berkeley (Nov 2016)

\end{thebibliography}

\appendix

\end{document}